\documentclass[conference]{IEEEtran}
\usepackage{amsmath,amssymb,amsthm,cite,graphicx,array}
\usepackage{graphicx}
\usepackage{float}
\usepackage{caption}
\usepackage{multicol}
\usepackage{mathrsfs}
\usepackage{amsmath}
\usepackage{amssymb}
\usepackage{amsthm}
\usepackage{multicol}
\usepackage{algorithm}
\usepackage[noend]{algpseudocode}
\usepackage{float}
\usepackage{placeins}
\usepackage{epstopdf}
\usepackage{multirow}
\usepackage{subfigure}
\usepackage{enumitem}
\usepackage[utf8]{inputenc}
\floatstyle{ruled}
\newfloat{algorithm}{tbp}{loa}
\providecommand{\algorithmname}{Algorithm}
\floatname{algorithm}{\protect\algorithmname}
\theoremstyle{remark}
\newtheorem{theorem}{Theorem}
\newtheorem{step}{Step}
\newtheorem{note}{Note}

\newtheorem{lemma}{Lemma}
\newtheorem{definition}{Definition}
\newtheorem{observation}{Observation}
\newtheorem{remark}{Remark}
\theoremstyle{remark}
\newtheorem{example}{Example}
\ifCLASSINFOpdf
\else
\fi

\title{Reduced Dimensional Optimal Vector Linear Index Codes for Index Coding Problems with Symmetric Neighboring and Consecutive Side-information}
\begin{document}
\author{Mahesh~Babu~Vaddi~and~B.~Sundar~Rajan\\ 
 Department of Electrical Communication Engineering, Indian Institute of Science, Bengaluru 560012, KA, India \\ E-mail:~\{vaddi,~bsrajan\}@iisc.ac.in }
\maketitle
\begin{abstract}
A single unicast index coding problem (SUICP) with symmetric neighboring and consecutive side-information (SNCS) has $K$ messages and $K$ receivers, the $k$th receiver $R_k$ wanting the $k$th message $x_k$ and having the side-information $\mathcal{K}_k=\{x_{k-U},\dots,x_{k-2},x_{k-1}\}\cup\{x_{k+1}, x_{k+2},\dots,x_{k+D}\}$. The single unicast index coding problem with symmetric neighboring and consecutive side-information, SUICP(SNCS), is motivated by topological interference management problems in wireless communication networks. 
Maleki, Cadambe and Jafar obtained the symmetric capacity of this SUICP(SNCS) and proposed optimal length codes by using Vandermonde matrices. In our earlier work, we gave optimal length $(U+1)$-dimensional vector linear index codes for SUICP(SNCS) satisfying some conditions on $K,D$ and $U$ \cite{VaR1}. In this paper, for SUICP(SNCS) with arbitrary $K,D$ and $U$, we construct optimal length $\frac{U+1}{\text{gcd}(K,D-U,U+1)}$-dimensional vector linear index codes. We prove that the constructed vector linear index code is of minimal dimension if $\text{gcd}(K-D+U,U+1)$ is equal to $\text{gcd}(K,D-U,U+1)$. The proposed construction gives optimal length scalar linear index codes for the SUICP(SNCS) if $(U+1)$ divides both $K$ and $D-U$. The proposed construction is independent of field size and works over every field. We give a low-complexity decoding for the SUICP(SNCS). By using the proposed decoding method, every receiver is able to decode its wanted message symbol by simply adding some index code symbols (broadcast symbols).		
\end{abstract}
\section{Introduction and Background}
\label{sec1}
\IEEEPARstart {A}{n} index coding problem, consists of a transmitter and a set of $m$ receivers, $R=\{R_0,R_1,\ldots,R_{m-1}\}$. The transmitter has a set of $n$ independent messages, $X=\{x_0,x_1,\ldots,x_{n-1}\}$. Each receiver, $R_k=(\mathcal{K}_k,\mathcal{W}_k)$, knows a subset of messages, $\mathcal{K}_k \subset X$, called its \textit{side-information}, and demands to know another subset of messages, $\mathcal{W}_k \subseteq \mathcal{K}_k^\mathsf{c}$, called its \textit{Want-set} or \textit{Demand-set}. The transmitter can take cognizance of the side-information of the receivers and broadcast coded messages, called the index code, over a noiseless channel. The problem of index coding with side-information was introduced by Birk and Kol \cite{BiK}. An index coding problem is called single unicast \cite{OnH} if the demand sets of the receivers are disjoint and the size of each demand set is one.


Index coding with side-information is motivated by wireless broadcasting applications. In applications like video-on-demand, during the initial transmission by the transmitter, each receiver may miss a part of data. A na\"{i}ve technique is to rebroadcast the entire data again to the receivers. This is an inefficient approach. The index coding problem, therefore aims at reducing the number of transmissions by the transmitter by intelligently using the data already available at the receivers.  



In a vector linear index code (VLIC) $x_k \in \mathbb{F}_q^{p_k},~ x_k=(x_{k,1},x_{k,2},\ldots,x_{k, p_k }),~x_{k,j} \in \mathbb{F}_q$ for 
$k \in [0:n-1]$ and $j \in [1:p_k]$ where $\mathbb{F}_q$ is a finite field with $q$ elements. In vector linear index coding setting, we refer $x_k \in \mathbb{F}_q^{p_k}$ as a message vector or a message and $x_{k,1},x_{k,2},\ldots,x_{k,p_k} \in \mathbb{F}_q$ as the message symbols. An index coding is a mapping defined as
\begin{align*}
\mathfrak{E}: \mathbb{F}^{p_0+p_1+\ldots+p_{n-1}}_q \rightarrow \mathbb{F}^N_q,
\end{align*}
where $N$ is the length of index code. The index code $\mathfrak{C}=\{(c_0,c_1,\ldots,c_{N-1}) \}$ is the collection of all images of the mapping $\mathfrak{E}$. We call the symbols $c_0$,$c_1,\ldots$,$c_{N-1}$ the code (broadcast) symbols, which are the symbols broadcasted by the transmitter. If $p_0=p_1=\cdots=p_{n-1}$, then the index code is called symmetric rate vector index code. If $p_0=p_1=\cdots=p_{n-1}=1$, then the index code is called scalar index code. The index coding problem is to design an index code such that the number of transmissions $N$ broadcasted by the transmitter is minimized and all the receivers get their wanted messages by using the index code symbols (broadcast symbols) broadcasted and their known side-information. 

Maleki, Cadambe and Jafar \cite{MCJ} found the capacity of SUICP(SNCS) with $K$ messages and $K$ receivers, each receiver has a total of $U+D$ side-information, corresponding to the $U$ messages before and $D$ messages after its desired message. In this setting, the $k$th receiver $R_k$ demands the message $x_k$ having the side-information
\begin{equation}
\label{antidote}
{\cal K}_k= \{x_{k-U},\dots,x_{k-2},x_{k-1}\}\cup\{x_{k+1}, x_{k+2},\dots,x_{k+D}\}.
\end{equation}

The symmetric capacity of this index coding problem setting is:
\begin{equation}
\label{capacity}
C=\left\{
                \begin{array}{ll}
                  {1 ~~~~~~~~~~~~ \mbox{if} ~~ U+D=K-1}\\
                  {\frac{U+1}{K-D+U}} ~~~ \mbox{if} ~~U+D\leq K-2. 
                  \end{array}
              \right.
\end{equation}
where $U,D \in$ $\mathbb{Z},$ $0 \leq U \leq D$.

In the one-sided side-information case, i.e., the cases where $U$ is zero, the $k$th receiver $R_k$ demands the message $x_k$ having the side-information,
\begin{equation}
\label{antidote1}
{\cal K}_k =\{x_{k+1}, x_{k+2},\dots,x_{k+D}\}, 
\end{equation}
\noindent
for which \eqref{capacity} reduces to
\begin{equation}
\label{capacity1}
C=\left\{
                \begin{array}{ll}
                  {1 ~~~~~~~~~~~~ \mbox{if} ~~ D=K-1}\\
                  {\frac{1}{K-D}} ~~~~~~~ \mbox{if} ~~D\leq K-2 
                  \end{array}
              \right.
\end{equation}
symbols per message.
 
Jafar \cite{TIM} established the relation between index coding problem and topological interference management problem. The capacity and optimal coding results in index coding can be used in corresponding topological interference management problems. The symmetric and neighboring side-information problems are motivated by topological interference management problems.

\subsection{Review of Adjacent Independent Row (AIR) matrices}
In \cite{VaR2}, we constructed binary matrices of size $m \times n (m\geq n)$ such that any $n$ adjacent rows of the matrix are linearly independent over every field. We refer these matrices as Adjacent Independent Row (AIR) matrices.

The matrix obtained by Algorithm \ref{algo2} is called the $(m,n)$ AIR matrix and it is denoted by $\mathbf{L}_{m\times n}.$ The general form of the $(m,n)$ AIR matrix is shown in Fig. \ref{fig1}. It consists of several submatrices (rectangular boxes) of different sizes as shown in Fig.\ref{fig1}. 
\begin{figure*}
\centering
\includegraphics[scale=0.38]{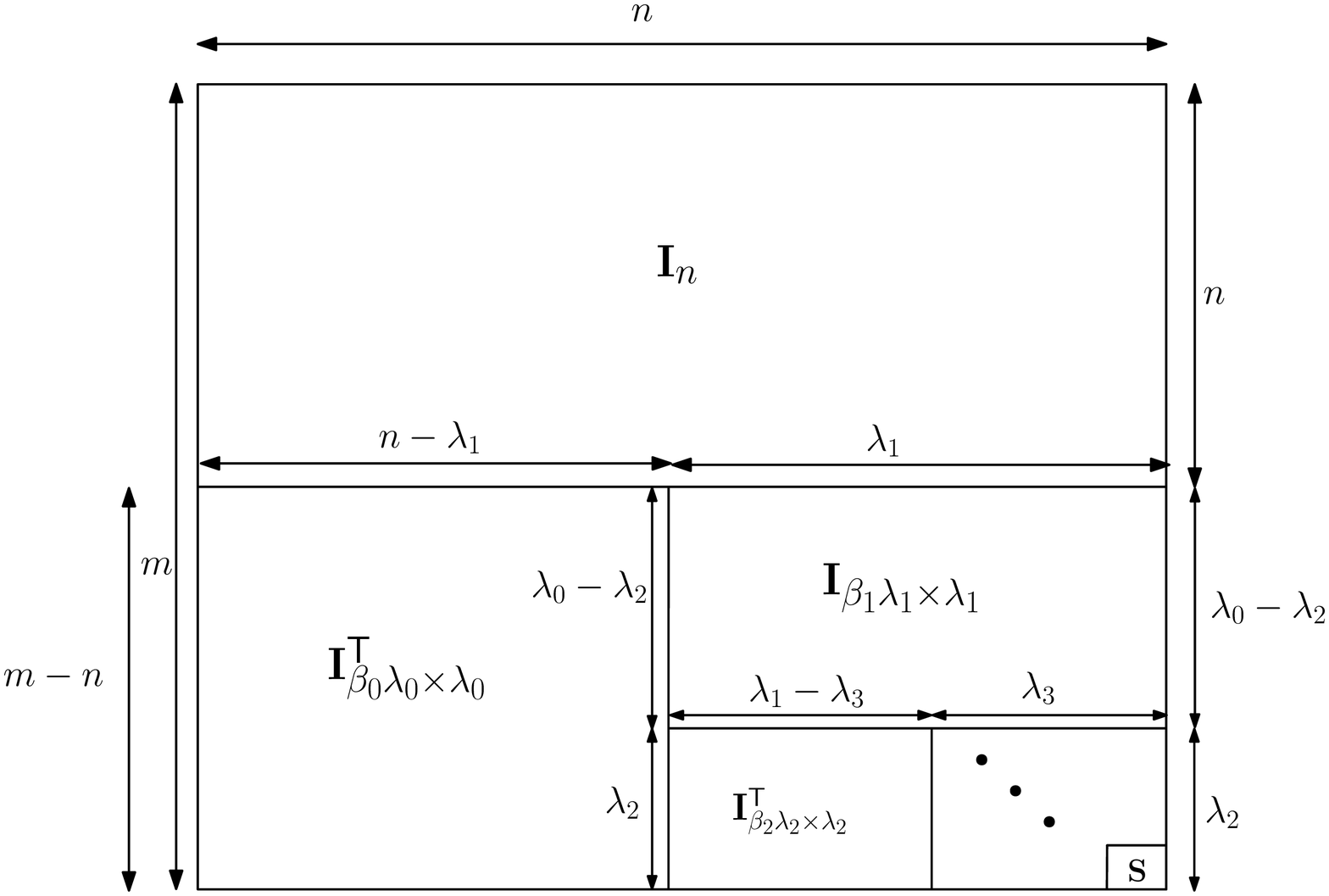}\\
~ $\mathbf{S}=\mathbf{I}_{\lambda_l \times \beta_l \lambda_l}$ if $l$ is even and ~$\mathbf{S}=\mathbf{I}_{\beta_l\lambda_l \times \lambda_l}$ otherwise.
\caption{AIR matrix of size $m \times n$.}
\label{fig1}
~ \\
\hrule
\end{figure*}
The description of the submatrices are as follows: Let $c$ and $d$ be two positive integers and $d$ divides $c$. The following matrix denoted by $\mathbf{I}_{c \times d}$ is a rectangular matrix.
\begin{align}
\label{rcmatrix}
\mathbf{I}_{c \times d}=\left.\left[\begin{array}{*{20}c}
   \mathbf{I}_d  \\
   \mathbf{I}_d  \\
   \vdots  \\
   \mathbf{I}_d 
   \end{array}\right]\right\rbrace \frac{c}{d}~\text{number~of}~ \mathbf{I}_d~\text{matrices}
\end{align}
and $\mathbf{I}_{d \times c}$ is the transpose of $\mathbf{I}_{c \times d}.$

Towards explaining the other quantities shown in the AIR matrix shown in Fig. \ref{fig1}, for a given $m$ and $n,$ let $\lambda_{-1}=n,\lambda_0=m-n$ and\begin{align}
\nonumber
n&=\beta_0 \lambda_0+\lambda_1, \nonumber \\
\lambda_0&=\beta_1\lambda_1+\lambda_2, \nonumber \\
\lambda_1&=\beta_2\lambda_2+\lambda_3, \nonumber \\
\lambda_2&=\beta_3\lambda_3+\lambda_4, \nonumber \\
&~~~~~~\vdots \nonumber \\
\lambda_i&=\beta_{i+1}\lambda_{i+1}+\lambda_{i+2}, \nonumber \\ 
&~~~~~~\vdots \nonumber \\ 
\lambda_{l-1}&=\beta_l\lambda_l.
\label{chain}
\end{align}
where $\lambda_{l+1}=0$ for some integer $l,$ $\lambda_i,\beta_i$ are positive integers and $\lambda_i < \lambda_{i-1}$ for $i=1,2,\ldots,l$. The number of submatrices in the AIR matrix is $l+2$ and the size of each submatrix is shown using $\lambda_i,\beta_i,$  $i \in [0:l].$

		\begin{algorithm}
		\caption{Algorithm to construct the AIR matrix $\mathbf{L}$ of size $m \times n$}
			\begin{algorithmic}[2]
				 \item Let $\mathbf{L}=m \times n$ blank unfilled matrix.
				\item [Step 1]~~~
				\begin{itemize}
				\item[\footnotesize{1.1:}] Let $m=qn+r$ for $r < n$.
				\item[\footnotesize{1.2:}] Use $\mathbf{I}_{qn \times n}$ to fill the first $qn$ rows of the unfilled part of $\mathbf{L}$.
				\item[\footnotesize{1.3:}] If $r=0$, Go to Step 3.
				\end{itemize}

				\item [Step 2]~~~
				\begin{itemize}
				\item[\footnotesize{2.1:}] Let $n=q^{\prime}r+r^{\prime}$ for $r^{\prime} < r$.
				\item[\footnotesize{2.2:}] Use $\mathbf{I}_{q^{\prime}r \times r}^{\mathsf{T}}$ to fill the first $q^{\prime}r$ columns of the unfilled part of $\mathbf{L}$.
			    \item[\footnotesize{2.3:}] If $r^{\prime}=0$, go to Step 3.	
				\item[\footnotesize{2.4:}] $m\leftarrow r$ and $n\leftarrow r^{\prime}$.
				\item[\footnotesize{2.5:}] Go to Step 1.
				\end{itemize}
				\item [Step 3] Exit.
		
			\end{algorithmic}
			\label{algo2}
		\end{algorithm}

In \cite{VaR2}, we gave an optimal length scalar linear index code for one-sided SUICP(SNCS) using AIR encoding matrices. In \cite{VaR1}, we constructed optimal length $(U+1)$ dimensional VLICs for two-sided SUICP(SNCS) satisfying some conditions on $K,D$ and $U$. The VLIC construction in \cite{VaR1} does not use AIR matrices. In \cite{VaR3}, we gave a low-complexity decoding for one-sided SUICP(SNCS) with AIR matrix as encoding matrix. The low complexity decoding method helps to identify a reduced set of side-information for each user with which the decoding can be carried out. By this method every receiver is able to decode its wanted message symbol by simply adding some broadcast symbols.	

\subsection{Contributions}

In a $b$-dimension vector index coding, the transmitter has to wait for $b$-realizations of a given message symbol to perform the index coding. In scalar index codes, the transmitter encodes each realization of the message symbols separately. In a delay critical environment like real time video streaming, it may not be desirable to wait for $b$-realizations of a given message symbol. Hence, in this paper, we focus on reducing the dimension of VLICs without compromising on optimal length of the code.

\begin{itemize}
\item For the two-sided SUICP(SNCS) with arbitrary $K,D$ and $U$, we construct optimal length $\frac{U+1}{\text{gcd}(K,D-U,U+1)}$ dimensional VLICs. The proposed construction is independent of field size and works over every field. 
\item We prove that the constructed VLICs are of minimal dimension if $\text{gcd}(K-D+U,U+1)$ is equal to $\text{gcd}(K,D-U,U+1)$. In other words, optimal VLICs of lesser dimension do not exist.
\item We give a low-complexity decoding for the two-sided SUICP(SNCS). By this method every receiver is able to decode its wanted message symbol by simply adding some broadcast symbols.	
\item We give generator matrices for the proposed VLICs by using AIR matrices.
\end{itemize}

The SUICP(SNCS) considered in this paper was referred as symmetric neighboring antidote multiple unicast index coding problem by Maleki, Cadambe and Jafar in \cite{MCJ}.

All the subscripts in this paper are to be considered $~\text{\textit{modulo}}~ K$. 

The paper is organized as follows. In Section \ref{sec2}, we give a construction of reduced dimension optimal VLICs for two-sided SUICP(SNCS). In Section \ref{sec3}, we give some special cases of the proposed construction.

\section{Construction of Reduced Dimension Vector Linear Index Codes}
\label{sec2}
For the $b$ dimensional vector linear index coding, we refer $x_k \in \mathbb{F}_q^b$ as a message vector and $x_{k,1},x_{k,2},\ldots,x_{k,b} \in \mathbb{F}_q$ as message symbols.

\begin{theorem}
\label{vcthm1}
Given arbitrary positive integers $K,D$ and $U$, consider the two-sided SUICP(SNCS) with $K$ messages $\{x_0,x_1,\cdots,x_{K-1}\}$ and the receivers being $K,$ and the receiver $R_k$ for $k \in [0:K-1]$ wanting the message $x_k$ and having the side-information given by 
\begin{align}
\label{thm1eq}
{\cal K}_k= \{x_{k-U},\dots,x_{k-2},x_{k-1}\}~\cup \{x_{k+1},x_{k+2},\dots,x_{k+D}\}.
\end{align}


Let 
\begin{align}
\label{thm1eq2}
u_a=\frac{U+1}{a},~~~~&\Delta_a=\frac{D-U}{a},~~~~ K_a=\frac{K}{a},
\end{align}
where
\begin{align}
\label{thm1eq1}
a=\text{gcd}(K,D-U,U+1).
\end{align}

Let $x_k=(x_{k,1},x_{k,2},\cdots,x_{k,u_a})$ be the message vector wanted by the $k$th receiver, where $x_k \in \mathbb{F}_q^{u_a}, x_{k,i} \in \mathbb{F}_q$ for every $k \in [0:K-1]$ and $i \in [1:u_a]$. 

Let
\begin{align}
\label{thm1eq6}
y_s=\sum_{i=1}^{u_a} \sum_{j=0}^{a-1}& x_{a(s+1-i)+j,i} ~~~~~~~~~s \in [0:K_a-1].
\end{align}

Let $\mathbf{L}$ be the AIR matrix of size $K_a \times (K_a-\Delta_a)$ and $L_s$ be the $s$th row of $\mathbf{L}$ for every $s \in [0:K_a-1]$. The optimal length $u_a$ dimensional VLIC for the two-sided SUICP(SNCS) with $K$ messages and $D$ side-information after and $U$ side-information before is given by 
\begin{align}
\label{thm1eq5}
[c_0~c_1~\ldots~c_{K_a-\Delta_a-1}]=\sum_{s=0}^{K_a-1}y_sL_s.
\end{align} 
\end{theorem}

\begin{proof}
We prove that every receiver $R_k$ for $k \in [0:K-1]$ decodes its wanted message vector $(x_{k,1},x_{k,2},\ldots,x_{k,u_a})$. 

Let $s=\left \lfloor \frac{k}{a} \right  \rfloor$. We have 
\begin{align}
\label{thm1rem}
k=as+ r~~\text{for~some}~ r \in [0:a-1]. 
\end{align}

Let
\begin{align}
\label{thm1eq3}
&z_{s,i}=\sum_{j=0}^{a-1} x_{as+j,i}, ~~~~s \in [0:K_a-1]~\text{and}~i \in [1:u_a].
\end{align}

We refer $z_{s,i}$ as extended message symbol and $z_s=(z_{s,1}~z_{s,2}~\ldots,z_{s,u_a})$ as extended message vector. The extended message symbol $z_{s,i}$ comprises of $'a'$ message symbols $x_{as,i},x_{as+1,i},\ldots,x_{as+a-1,i}$. Each of the $Ku_a$ message symbols $x_{k,i}$ for $k \in [0:K-1]$ and $i \in [1:u_a]$ appear exactly once in $K_au_a$ extended message symbols $z_{s,i}$ for $s \in [0:K_a-1]$ and $i \in [1:u_a]$.

From \eqref{thm1eq6} and \eqref{thm1eq3}, we have
\begin{align}
\label{thm1eq4}
y_s=\sum_{i=1}^{u_a} z_{s+1-i,i}~~\text{for}~~s \in [0:K_a-1].
\end{align}

The symbol $y_s$ comprises of $u_a$ extended message symbols $z_{s,1},z_{s-1,2},$ $\ldots,z_{s-u_a-1,u_a}$. The $K_au_a$ symbols $z_{s,i}$ for $s \in [0:K_a-1]$ and $i \in [1:u_a]$ appear exactly once in $K_a$ symbols $y_s$ for $s \in [0:K_a-1]$. By combining \eqref{thm1eq3} and \eqref{thm1eq4}, the symbol $y_s$ comprises of $au_a=U+1$ message symbols corresponding to $U+1$ neighboring receivers. 
\begin{itemize}
\item The symbol $y_s$ comprises of $(U+1)$ message symbols one from each of the message vectors $x_t$ for every $t \in [as-(U+1)+a:as+a-1]$.
\item The symbol $y_{s+1}$ comprises of $U+1$ message symbols one from each of the message vectors $x_t$ for every $t \in [as-(U+1)+2a:as+2a-1]$.
\item The symbol $y_{s+\Delta_a}$ comprises of $U+1$ message symbols one from each of the message vectors $x_t$ for every $t \in [as-(U+1)+a(\Delta_a+1):as+a(\Delta_a+1)-1]$.
\item The symbol vector $[y_s~y_{s+1}~\ldots~y_{s+\Delta_a}]$ comprises of $(U+1)\Delta_a$ message symbols from the message vectors $x_t$ for every $t \in [as-(U+1)+a:as+a(\Delta_a+1)-1]=[\underbrace{as-(U+1)+a}_{=g}:\underbrace{as+D-U+a-1}_{=h}]$. \\Note that $h-g=D$.
\end{itemize}

From the alignment of the symbols $y_s,y_{s+1},\ldots,y_{s+\Delta_a}$ and the side-information given in \eqref{thm1eq}, we can conclude the following:
\begin{observation}
\label{obs1}
For every receiver $R_k$ whose wanted message symbol is present in $y_s$, every other message symbol in $y_s$ is in side-information of $R_k$.
\end{observation}
\begin{observation}
\label{obs2}
For every receiver $R_k$ whose wanted message is present in $y_s$, every message symbol in $y_{s+1},y_{s+2},\ldots,y_{s+\Delta_a}$ which does not belong to the message vector $x_k$ is in side-information of $R_k$. That is, for every receiver $R_k$ whose wanted message is present in $y_s$, every extended message symbol in $y_{s+1},y_{s+2},\ldots,y_{s+\Delta_a}$ which does not belong to the message vector $z_s$ is in side-information of $R_k$ (we refer the extended message symbol $z_{t,i}=\sum_{j=0}^{a-1} x_{at+j,i}$ is in the side-information of $R_k$ if every message symbol $x_{at+j,i}$ for $j \in [0:a-1]$ is in the side-information of $R_k$). 
\end{observation}

Consider the one-sided SUICP(SNCS) with $K_a$ messages $\{y_0,y_1,\cdots,y_{K_a-1}\}$, $K_a$ receivers and the receiver wanting the message $y_s$ and having the side-information given by ${\cal K}_s= \{y_{s+1},y_{s+2},\dots,y_{s+\Delta_a}\},$ for every $s \in [0:K_a-1]$. Let $\mathfrak{C}$ be the optimal length index code given by \eqref{thm1eq5}. In \cite{VaR2}, we proved that the code $\mathfrak{C}$ enables the decoding of $K_a$ message $y_0,y_1,\cdots,y_{K_a-1}$. That is, for every $s \in [0:K_a-1]$, with linear decoding there exists a linear combination of broadcast symbols $c_0,c_1, \cdots, c_{K_a-\Delta_a-1}$ to get the sum of the form  
\begin{align}
\label{vcsum}
S_s=y_s+b_s^{(1)}y_{s+a_s^{(1)}}+b_s^{(2)}y_{s+a_s^{(2)}}+\cdots+
b_s^{(d)}y_{s+a_s^{(d)}} 
\end{align} 
where $b_s^{(1)},b_s^{(2)},\ldots,b_s^{(d)} \in \mathbb{F}_q,d \leq \Delta_a$ and $1 \leq a_s^{(1)} <a_s^{(2)}< \ldots <a_s^{(d)} \leq \Delta_a$.

In the given two-sided SUICP(SNCS), we have $k \in [0:K-1]$, $s=\left \lfloor \frac{k}{a} \right  \rfloor \in [0:K_a-1].$ From \eqref{thm1eq3}, the message symbol $x_{k,i}$ for $i \in [1:u_a]$ is present in the extended message symbol $z_{s,i}$. For every message symbol $x_{k,i}$ for $i \in [1:u_a]$, the decoding is performed as given below.
\begin{itemize}
\item $R_k$ first decodes the extended message symbol $z_{s,i}$ for $i \in [1:u_a]$, where the message symbol $x_{k,i}$ is present.
\item In $z_{s,i}=\sum_{j=0}^{a-1} x_{as+j,i}$, every message symbol present is in the side-information of $R_k$. Hence, $R_k$ decodes its wanted message symbol $x_{k,i}$ from $z_{s,i}$.
\end{itemize}
\begin{step}
\label{step1}
Decoding of $x_{k,u_a}$
\end{step}

The message symbol $x_{k,u_a}$ is present in $z_{s,u_a}$. Receiver $R_k$ first decodes the extended message symbol $z_{s,u_a}$ as follows. From \eqref{thm1eq4}, we have 
\begin{align*}
y_{s+u_a-1}&=\sum_{i=1}^{u_a} z_{s+u_a-i,i}=z_{s,u_a}+\sum_{i=1}^{u_a-1} z_{s+u_a-i,i}.
\end{align*}
Hence, the extended message symbol $z_{s,u_a}$ is present in $y_{s+u_a-1}$. The symbol $y_{s+u_a-1}$ can be decoded from $S_{s+u_a-1}$. The sum $S_{s+u_a-1}$ obtained from \eqref{vcsum} is given in \eqref{vcstep10} and \eqref{vcstep1}.
\begin{figure*}[ht]
\begin{align}
\label{vcstep10}
S_{s+u_a-1}=y_{s+u_a-1}+\underbrace{\sum_{t=1}^d b_{s+u_a-1}^{(t)}y_{s+u_a-1+a_{s+u_a-1}^{(t)}}}_{\text{side~information~to}~R_k}.
\end{align}
\begin{align}
\label{vcstep1}
S_{s+u_a-1}=\underbrace{\underbrace{z_{s,u_a}}_{\text{wanted extended message symbol to} R_k} +\underbrace{\sum_{i=1}^{u_a-1} z_{s+u_a-i,i}}_{\text{side~information~to~receiver~}R_k}}_{y_{s+u_a-1}}+\underbrace{\sum_{t=1}^d b_{s+u_a-1}^{(t)}y_{s+u_a-1+a_{s+u_a-1}^{(t)}}}_{\text{side~information~to}~R_k}.
\end{align}
\end{figure*}
In $S_{s+u_a-1}$, we have 
\begin{align}
\label{vcstep1eq1}
1 \leq a_{s+u_a-1}^{(1)}<a_{s+u_a-1}^{(2)}< \cdots <a_{s+u_a-1}^{(d)}\leq \Delta_a. 
\end{align}
From \eqref{thm1eq4} and \eqref{vcstep1eq1}, we have 
\begin{align*}
y_{s+u_a-1+t}=\sum_{i=1}^{u_a}& z_{s+u_a+t-i,i}~\text{for}~s \in [0:K_a-1], \\&t \in \{a_{s+u_a-1}^{(1)},a_{s+u_a-1}^{(2)},\ldots,a_{s+u_a-1}^{(d)}\},
\end{align*}
and in $y_{s+u_a-1+t}$, the extended message symbols belonging to $z_s$ are not present. Hence, in $S_{s+u_a-1}$, only one extended message symbol is present which belongs to $z_s.$ From Observation \ref{obs2}, the symbols $y_{s+u_a-1+t}$ for every $t \in \{a_{s+u_a-1}^{(1)},a_{s+u_a-1}^{(2)},\ldots,a_{s+u_a-1}^{(d)}\}$ are in the side-information of $R_k$. From Observation \ref{obs1}, every message symbol present in $y_{s+u_a-1}$ is in side-information of $R_k$ except $x_{k,u_a}$. Hence, by using $S_{s+u_a-1}$, receiver $R_k$ decodes its wanted extended message symbol $z_{s,u_a}$.
From \eqref{thm1rem} and \eqref{thm1eq3}, we have 
\begin{align*}
z_{s,u_a}&=\sum_{j=0}^{a-1} x_{as+j,u_a} \\&=\sum_{j=0}^{a-1} x_{k-r+j,u_a}~\text{where}~r \in [0:a-1]\\&=x_{k,u_a}+\underbrace{\sum_{j=0,j \neq r}^{a-1} x_{k-r+j,u_a}}_{\text{side~information~to}~R_k}.
\end{align*}

Hence, $R_k$ decodes its wanted message $x_{k,u_a}$ by using $z_{s,u_a}$ (Note that if $S_{s+u_a-1}$ comprises $h \geq 2$ extended message symbols for some integer $h$ which belong to $z_s$ and if $R_k$ yet to decode $c$ of them for some integer $c$ such that $2 \leq c \leq h$, then the $c$ extended message symbols interfere and $R_k$ can not decode any of the $c$ message symbols).


\begin{step}
\label{step2}
Decoding of $x_{k,u_a-1}$
\end{step}

The message symbol $x_{k,u_a-1}$ is present in the extended message symbol $z_{s,u_a-1}$. Receiver $R_k$ first decodes the extended message symbol $z_{s,u_a-1}$ as follows. From \eqref{thm1eq4}, we have 
\begin{align*}
y_{s+u_a-2}&=\sum_{i=1}^{u_a} z_{s+u_a-1-i,i}\\&=z_{s,u_a-1}+\sum_{i=1,i \neq u_a-1}^{u_a} z_{s+u_a-1-i,i}.
\end{align*}
Hence, the extended message symbol $z_{s,u_a-1}$ is present in $y_{s+u_a-2}$. Receiver $R_k$ uses the sum $S_{s+u_a-2}$ to decode $z_{s,u_a-1}$. The sum $S_{s+u_a-1}$ obtained from \eqref{vcsum} is given in \eqref{vcstep20} and \eqref{vcstep2}.
\begin{figure*}
\begin{align}
\label{vcstep20}
S_{s+u_a-2}=y_{s+u_a-2}+b_{s+u_a-2}^{(1)}y_{s+u_a-2+a_{s+u_a-2}^{(1)}}+\underbrace{\sum_{t=2}^d b_{s+u_a-2}^{(t)}y_{s+u_a-2+a_{s+u_a-2}^{(t)}}}_{\text{side~information~to}~R_k}.
\end{align}
\begin{align}
\label{vcstep2}
\nonumber
S_{s+u_a-2}=&\underbrace{\underbrace{z_{s,u_a-1}}_{\text{wanted~extended~message~symbol~to}~R_k}+\underbrace{\sum_{i=1,i \neq u_a-1}^{u_a} z_{s+u_a-1-i,i}
}_{\text{side~information~to~receiver~}R_k}}_{y_{s+u_a-2}}\\& +
\underbrace{\underbrace{b_{s+u_a-2}^{(1)}z_{s-1+a_{s+u_a-2}^{(1)},u_a}}_{\text{interference~from}~z_s~\text{if}~a_{s+u_a-2}^{(1)}=1}+\underbrace{b_{s+u_a-2}^{(1)}\sum_{i=1}^{u_a-1} z_{s+u_a-1+a_{s+u_a-2}^{(1)}-i,i}}_{\text{side~information~to}~R_k}}_{y_{s+u_a-2+a_{s+u_a-2}^{(1)}}}+\underbrace{\sum_{t=2}^d b_{s+u_a-2}^{(t)}y_{s+u_a-2+a_{s+u_a-2}^{(t)}}}_{\text{side~information~to}~R_k}.
\end{align}
\end{figure*}

In \eqref{vcstep20}, we have 
\begin{align}
\label{vcstep2eq1}
1 \leq a_{s+u_a-2}^{(1)}<a_{s+u_a-2}^{(2)}< \cdots <a_{s+u_a-2}^{(d)}\leq \Delta_a.
\end{align}

Depending on the value of $a_{s+u_a-2}^{(1)}$, sum $S_{s+u_a-2}$ comprises of either one extended message symbol or two extended message symbols belonging to $z_s~(z_{s,1},z_{s,2},\ldots,z_{s,u_a-1},z_{s,u_a})$. If $a_{s+u_a-2}^{(1)}>1$, then $S_{s+u_a-2}$ comprises only one extended message symbol ($z_{s,u_a-1}$) belonging to $z_s$. If $a_{s+u_a-2}^{(1)}=1$, $S_{s+u_a-2}$ comprises two extended message symbols ($z_{s,u_a},z_{s,u_a-1}$) belonging to $z_s.$ Receiver $R_k$ has already decoded the extended message symbol $z_{s,u_a}$ in Step \ref{step1}. From Observation \ref{obs2}, the symbols $y_{s+u_a-2+t}$ for every $t \in \{a_{s+u_a-2}^{(1)},a_{s+u_a-2}^{(2)},\ldots,a_{s+u_a-2}^{(d)}\}$ are in the side-information of $R_k$. From Observation \ref{obs1}, every message symbol present in $y_{s+u_a-2}$ is in side-information of $R_k$ except $x_{k,u_a-1}$. Hence, by using $S_{s+u_a-2}$, receiver $R_k$ decodes its wanted extended message symbol $z_{s,u_a-1}$. From \eqref{thm1rem} and \eqref{thm1eq3}, we have 
\begin{align*}
z_{s,u_a-1}&=\sum_{j=0}^{a-1} x_{as+j,u_a-1} \\&=\sum_{j=0}^{a-1} x_{k-r+j,u_a-1}~\text{where}~r \in [0:a-1]\\&=x_{k,u_a-1}+\underbrace{\sum_{j=0,j \neq r}^{a-1} x_{k-r+j,u_a-1}}_{\text{side~information~to}~R_k}.
\end{align*}

Hence, $R_k$ decodes its wanted message $x_{k,u_a-1}$ by using $z_{s,u_a-1}$.
%

\textit{Step l.}  Decoding of $x_{k,u_a+1-l}$ for $l \in [3:u_a]$

The message symbol $x_{k,u_a+1-l}$ is present in the extended message symbol $z_{s,u_a+1-l}$. Receiver $R_k$ first decodes the extended message symbol $z_{s,u_a+1-l}$ as follows. From \eqref{thm1eq4}, the extended message symbol $z_{s,u_a+1-l}$ is present in $y_{s+u_a-l}$. The symbol $y_{s+u_a-l}$ can be decoded from $S_{s+u_a-l}$. Define $A_j=\{a_{s+u_a-l}^{(j)}+1,a_{s+u_a-l}^{(j)}+2,\cdots,u_a\}$ if $u_a > a_{s+u_a-l}^{(j)}$, else $A_j=\phi$ for every $j \in [1:l-1]$. Let $1_A$ be the indicator function such that $1_A(x)=1$ if $x \in A$, else it is zero. The sum $S_{s+u_a-l}$ obtained from \eqref{vcsum} is given in \eqref{vcstepl0} and \eqref{vcstepl}.
\begin{figure*}[ht]
\begin{align}
\label{vcstepl0}
S_{s+u_a-l}&=y_{s+u_a-l}+\sum_{t=1}^{l-1}b_{s+u_a-l}^{(t)}y_{s+u_a-l+a_{s+u_a-l}^{(t)}}++\underbrace{\sum_{t=l}^d b_{s+u_a-l}^{(t)}y_{s+u_a-l+a_{s+u_a-l}^{(t)}}}_{\text{side~information~to}~R_k}.
\end{align}
\begin{align}
\nonumber
S_{s+u_a-l}=&\underbrace{\underbrace{z_{s,u_a+1-l}}_{\text{wanted~extended~message~symbol~to}~R_k}+ \underbrace{\sum\limits_{i=1, i\neq u_a+1-l}^{u_a} z_{s+u_a+1-l-i,i}
}_{\text{side~information~to}~R_k}}_{y_{s+u_a-l}}+\underbrace{\sum_{t=l}^{d}b_{s+u_a-l}^{(t)}y_{s+u_a-l+a_{s+u_a-l}^{(t)}}}_{\text{side~information~to}~R_k}\\&
+ \sum_{t=1}^{l-1}\underbrace{\underbrace{b_{s+u_a-l}^{(t)}1_{A_t}(l)z_{s,u_a+1-l+a_{s+u_a-l}^{(t)}}}_{\text{Interference~from}~z_s}+\underbrace{\sum_{t=1}^{l-1}b_{s+u_a-l}^{(t)}\left(y_{s+u_a-l+a_{s+u_a-l}^{(t)}}-1_{A_{t}}(l)z_{s,u_a+1-l+a_{s+u_a-l}^{(t)}}\right)}_{\text{side~information~to}~R_k}}_{y_{s+u_a-l+a_{s+u_a-l}^{(t)}}}.
\label{vcstepl}
\end{align}
\end{figure*}
In this sum $S_{s+u_a-l}$, $z_{s,u_a+1-l}$ is the required extended message symbol and $1_{A_j}(l)b_{s+u_a-l}^{(j)}z_{s,u_a+1-l+a_{s+u_a-l}^{(j)}}$ for $j \in [1:l-1]$ is the interference to the receiver $R_k$ from its wanted extended message vector $z_s$. Receiver $R_k$ already knows the $l-1$ extended message symbols $\{z_{s,u_a}, z_{s,u_a-1},\cdots,z_{s,u_a+2-l}\}$ and thus the interference from the extended message vector $z_s$ in the sum $S_{s+u_a-l}$ can be cancelled. From Observation \ref{obs2}, the symbols $y_{s+u_a-l+t}$ for every $t \in \{a_{s+u_a-l}^{(1)},a_{s+u_a-l}^{(2)},\ldots,a_{s+u_a-l}^{(d)}\}$ are in the side-information of $R_k$. From Observation \ref{obs1}, every message symbol present in $y_{s+u_a-l}$ is in side-information of $R_k$ except $x_{k,u_a+1-l}$. Hence, by using $S_{s+u_a-l}$, receiver $R_k$ decodes its wanted extended message symbol $z_{s,u_a+1-l}$. From \eqref{thm1rem} and \eqref{thm1eq3}, we have 
\begin{align*}
z_{s,u_a-l}&=\sum_{j=0}^{a-1} x_{as+j,u_a-l}=\sum_{j=0}^{a-1} x_{k-r+j,u_a-l}\\&=x_{k,u_a-l}+\underbrace{\sum_{j=0,j \neq r}^{a-1} x_{k-r+j,u_a-l}}_{\text{side~information~to}~R_k}.
\end{align*}

Hence, $R_k$ decodes its wanted message $x_{k,u_a-l}$ by using $z_{s,u_a-l}.$

The decoding continues until the receiver $R_k$ has decoded its $u_a$ wanted extended message symbols $z_{s,1},z_{s,2},\cdots,z_{s,u_a}$. The receiver $R_k$ decodes the extended message vector $z_s$ successively in the following order:
\begin{align*}
z_{s,u_a}\rightarrow z_{s,u_a-1}\rightarrow z_{s,u_a-2}\rightarrow \ldots \rightarrow z_{s,2}\rightarrow z_{s,1}.
\end{align*}  

Hence, The receiver $R_k$ decodes the wanted message $x_k$ successively in the following order:
\begin{align*}
x_{k,u_a}\rightarrow x_{k,u_a-1}\rightarrow x_{k,u_a-2}\rightarrow \ldots \rightarrow x_{k,2}\rightarrow x_{k,1}.
\end{align*}  

The extended message symbols wanted by receiver $R_k$ are decoded from $S_i$ as given in Table \ref{vctable1}. 
\begin{table*}[ht]
\centering
\begin{tabular}{|c|c|c|c|c|c|c|}
\hline
\textbf{$x_{k,u_a}$} & $x_{k,u_a-1}$ & $\ldots$ & $x_{k,u_a+1-l}$ & $\ldots$& $x_{k,2}$ & $x_{k,1}$\\
\hline
\textbf{$z_{s,u_a}$} & $z_{s,u_a-1}$ & $\ldots$ & $z_{s,u_a+1-l}$ & $\ldots$& $z_{s,2}$ & $z_{s,1}$\\
\hline
\textbf{$S_{s+u_a-1}$} & $S_{s+u_a-2}$ & $\ldots$& $S_{s+u_a-l}$ & $\ldots$ & $S_{s+1}$& $S_s$\\
\hline 
\end{tabular}
\caption{Sequential decoding of message vector from VLIC}
\label{vctable1}
\end{table*}

The number of broadcast symbols in the $u_a$ dimensional VLIC $\mathfrak{C}$ is $K_a-\Delta_a$. The rate achieved by $\mathfrak{C}$ is 
\begin{align*}
\frac{u_a}{K_a-\Delta_a}=\frac{\frac{U+1}{a}}{\frac{K}{a}-\frac{D-U}{a}}=\frac{U+1}{K-D+U}.
\end{align*}
This is equal to the capacity mentioned in \eqref{capacity}. Hence, the constructed VLICs are capacity achieving. 
\end{proof}


\begin{remark}
The decoding procedure given in Theorem \ref{vcthm1} uses successive interference cancellation. That is, in the decoding, the receiver $R_k$ decode its wanted message symbols $(x_{k,1},x_{k,2},\ldots,x_{k,u_a})$ one at a time starting from $x_{k,u_a}$ and proceeds to decode $$x_{k,u_a-1}\rightarrow x_{k,u_a-2} \rightarrow \ldots \rightarrow x_{k,1}$$ after cancelling the interference from already decoded message symbols.
\end{remark}

In the Lemma \ref{lemmaminimal} given below, we prove that the constructed optimal VLICs in Theorem \ref{vcthm1} are minimal dimensional 
if $\text{gcd}(K,D-U,U+1)$ is equal to $\text{gcd}(K-D+U,U+1)$. 
\begin{lemma}
\label{lemmaminimal}
If $\text{gcd}(K,D-U,U+1)=\text{gcd}(K-D+U,U+1)$, then the constructed VLICs in Theorem \ref{vcthm1} are minimal dimensional.
\end{lemma}
\begin{proof}
Let $a=\frac{U+1}{\text{gcd}(K-D+U,U+1)}, b=\frac{K-D+U}{\text{gcd}(K-D+U,U+1)}$. We have $\text{gcd}(a,b)=1$. For SUICP(SNCS), we have
\begin{align*}
C=\frac{U+1}{K-D+U}=\frac{\frac{U+1}{\text{gcd}(K-D+U,U+1)}}{\frac{K-D+U}{\text{gcd}(K-D+U,U+1)}}= \frac{a}{b}.
\end{align*}
If $a=1$, then the index code is a scalar linear index code and it is the minimal dimensional index code. Hence, with out loss of generality, we assume $a>1$. Consider there exist an $a-t$ dimensional capacity achieving VLIC for some $t \in [1:a-1]$. This $a-t$ dimensional VLIC is obtained by mapping $(a-t)K$ message symbols into $b-s$ broadcast symbols for some $s \in [1:b-1]$. This code is capacity achieving and hence we have,
\begin{align}
\label{lneweq1}
\frac{a-t}{b-s}=\frac{a}{b}
\end{align}
 which gives $as-bt=0.$ From the extended Euclid algorithm and Bezout coefficients, the equation of the form $as-bt=0$ has solutions of the form $s=k\frac{b}{\text{gcd}(a,b)}$ and $t=k\frac{a}{\text{gcd}(a,b)}$ for $k \in Z_{>0}$. We have $\text{gcd}(a,b)=1$ and hence $s=kb \geq b$ and $t=ka \geq a$. This is a contradiction. Hence, the dimension used in Theorem \ref{vcthm1} $$a=\frac{U+1}{\text{gcd}(K-D+U,U+1)}=\frac{U+1}{\text{gcd}(K,D-U,U+1)}$$ is minimal.
\end{proof}
For all SUICP(SCNC) not satisfying the condition $\text{gcd}(K,D-U,U+1)=\text{gcd}(K-D+U,U+1)$, we conjecture that the constructed optimal VLICs in Theorem \ref{vcthm1} are minimal dimensional.
\subsection{Low Complexity Decoding}
In \cite{VaR3}, we gave a low-complexity decoding for one-sided SUICP(SNCS) with AIR matrix as encoding matrix. The low complexity decoding method helps to identify a reduced set of side-information for each user with which the decoding can be carried out. By this method every receiver is able to decode its wanted message symbol by simply adding some broadcast symbols. The low complexity decoding of one-sided SUICP(SNCS) with $K_a$ messages and $\Delta_a$ side-information explicitly gives the values of $a_t^{(d)}$ and $b_t^{(d)}$ in \eqref{vcsum} for $t \in [0:K_a-1]$ and $d \in [1:\Delta_a]$. Hence, from the proof of Theorem \ref{vcthm1} we get a  low complexity decoding for two-sided SUICP(SNCS) with arbitrary $K,D$ and $U$.	

The following examples illustrate the construction of optimal VLICs and low complexity decoding for two-sided  SUICP(SNCS).


\begin{example}
\label{vcex1}
Consider the two-sided SUICP(SNCS) with $K=8, D=2, U=1$. For this SUICP(SNCS), $\text{gcd}(K,D-U,U+1)=a=1$. The optimal length index code is a two dimensional index code. For this two-sided SUICP(SNCS), we have $K_a=8$ and $\Delta_a=1$. The AIR matrix of size $\mathbf{L}_{8 \times 7}$ is given below.
\begin{figure}[ht]
\arraycolsep=1pt
{
$$\mathbf{L}_{8 \times 7}=\left[\begin{array}{*{20}c}
   1 & 0 & 0 & 0 & 0 & 0 & 0\\
   0 & 1 & 0 & 0 & 0 & 0 & 0\\
   0 & 0 & 1 & 0 & 0 & 0 & 0\\
   0 & 0 & 0 & 1 & 0 & 0 & 0\\
   0 & 0 & 0 & 0 & 1 & 0 & 0\\
   0 & 0 & 0 & 0 & 0 & 1 & 0\\
   0 & 0 & 0 & 0 & 0 & 0 & 1\\
 1 & 1 & 1 & 1 &1 & 1 & 1\\
   \end{array}\right]$$
}

\end{figure}

The scalar linear index code for the one-sided SUICP(SNCS) with $K_a=8$  and $\Delta_a=1$ is given by 
\begin{align*}
\mathfrak{C}=\{&y_0+y_7, ~~~y_1+y_7,~~~y_2+y_7, ~~~y_3+y_7, \\& y_4+y_7, ~~~y_5+y_7, ~~~y_6+y_7\}. 
\end{align*}

The VLIC for the given two-sided SUICP(SNCS) is obtained by replacing $y_s$ in $\mathfrak{C}$ with $x_{s,1}+x_{s-1,2}$ for $s \in [0:7]$. The VLIC is given by
\begin{align*}
\mathfrak{C}^{(2s)}=\{&x_{0,1}+x_{7,2}+x_{7,1}+x_{6,2}, ~x_{1,1}+x_{0,2}+x_{7,1}+x_{6,2}, \\&x_{2,1}+x_{1,2}+x_{7,1}+x_{6,2},~x_{3,1}+x_{2,2}+x_{7,1}+x_{6,2}, \\&x_{4,1}+x_{3,2}+x_{7,1}+x_{6,2}, ~x_{5,1}+x_{4,2}+x_{7,1}+x_{6,2}, \\&x_{6,1}+x_{5,2}+x_{7,1}+x_{6,2}\}. 
\end{align*}

Let $\tau_s$ be the set of broadcast symbols required to decode $y_s$ from $\mathfrak{C}$. The values of $\tau_s$ and $S_s$ are given in Table \ref{vctableex122}.
\begin{table}[ht]
\centering
\begin{tabular}{|c|c|c|}
\hline
{$\mathcal{W}_s$} & $\tau_s$ & $S_s$ \\
\hline
$y_0$ &$y_0+y_7,y_1+y_7$& $S_0=y_0+y_1$ \\
\hline
$y_1$&$y_1+y_7,y_2+y_7$&  $S_1=y_1+y_2$ \\
\hline 
$y_2$ &$y_2+y_7,y_3+y_7$&$S_2=y_2+y_3$  \\
\hline
$y_3$ &$y_3+y_7,y_4+y_7$&$S_3=y_3+y_4$  \\
\hline
$y_4$ &$y_4+y_7,y_5+y_7$& $S_4=y_4+y_5$  \\
\hline
$y_5$ &$y_5+y_7,y_6+y_7$&$S_5=y_5+y_6$ \\
\hline
$y_6$ &$y_6+y_7$&$S_6=y_6+y_7$ \\
\hline
$y_7$ &$y_0+y_7$&$S_7=y_7+y_0 $\\
\hline 
\end{tabular}
\caption{Decoding at each receiver for one-sided SUICP(SNCS) with $K=8,\Delta_a=1$}
\label{vctableex122}
\end{table}

Receiver $R_k$ wants to decode $x_{k,1}$ and $x_{k,2}$ for $k \in [0:7]$. $R_k$ first decode $x_{k,2}$ from $S_{s+1}$, then decodes $x_{k,1}$ from $S_s$ for every $k \in [0:7]$. The decoding procedure at the receivers $R_0$, $R_1\ldots R_7$ is shown in Table \ref{vctableex121}.

\begin{table}[ht]
\centering
\begin{tabular}{|c|c|c|c|}
\hline
$R_k$ & {$\mathcal{W}_k$} & Sum from which $x_{k,i}$ is decoded \\
\hline
$R_0$ & $x_{0,2}$&  $S_1=x_{1,1}+x_{0,2}+x_{2,1}+x_{1,2}$ \\
\hline 
$R_0$ & $x_{0,1}$ &$S_0=x_{0,1}+x_{7,2}+x_{1,1}+x_{0,2}$  \\
\hline
$R_1$ & $x_{1,2}$&  $S_2=x_{2,1}+x_{1,2}+x_{3,1}+x_{2,2}$ \\
\hline 
$R_1$ & $x_{1,1}$ &$S_1=x_{1,1}+x_{0,2}+x_{2,1}+x_{1,2}$  \\
\hline
$R_2$ & $x_{2,2}$&  $S_3=x_{3,1}+x_{2,2}+x_{4,1}+x_{3,2}$ \\
\hline 
$R_2$ & $x_{2,1}$ &$S_2=x_{2,1}+x_{1,2}+x_{3,1}+x_{2,2}$  \\
\hline
$R_3$ & $x_{3,2}$&  $S_4=x_{4,1}+x_{3,2}+x_{5,1}+x_{4,2}$ \\
\hline 
$R_3$ & $x_{3,1}$ &$S_3=x_{3,1}+x_{2,2}+x_{4,1}+x_{3,2}$  \\
\hline
$R_4$ & $x_{4,2}$&  $S_5=x_{5,1}+x_{4,2}+x_{6,1}+x_{5,2}$ \\
\hline 
$R_4$ & $x_{4,1}$ &$S_4=x_{4,1}+x_{3,2}+x_{5,1}+x_{4,2}$  \\
\hline
$R_5$ & $x_{5,2}$&  $S_6=x_{6,1}+x_{5,2}+x_{7,1}+x_{6,2}$ \\
\hline 
$R_5$ & $x_{5,1}$ &$S_5=x_{5,1}+x_{4,2}+x_{6,1}+x_{5,2}$  \\
\hline
$R_6$ & $x_{6,2}$&  $S_7=x_{7,1}+x_{6,2}+x_{0,1}+x_{7,2}$ \\
\hline 
$R_6$ & $x_{6,1}$ &$S_6=x_{6,1}+x_{5,2}+x_{7,1}+x_{6,2}$  \\
\hline
$R_7$ & $x_{7,2}$&  $S_0=x_{0,1}+x_{7,2}+x_{1,1}+x_{0,2}$ \\
\hline 
$R_7$ & $x_{7,1}$ &$S_{7}=x_{7,1}+x_{6,2}+x_{0,1}+x_{7,2}$  \\
\hline
\end{tabular}
\caption{Decoding at each receiver for the two-sided SUICP(SNCS) with $K=8,D=2$ and $U=1$}
\label{vctableex121}
\end{table}

The capacity of the given two-sided SUICP(SNCS) is $\frac{2}{7}$. The VLIC $\mathfrak{C}$ uses seven broadcast symbols to transmit two message symbols per receiver. The rate achieved by $\mathfrak{C}$ is $\frac{2}{7}$. Hence, $\mathfrak{C}^{(2)}$ is an optimal length index code. 
\end{example}

\begin{example}
\label{vcex2}
Consider the two-sided SUICP(SNCS) with $K=22, D=7, U=3$. For this SUICP(SNCS), $\text{gcd}(K,D-U,U+1)=a=2$. The optimal length index code is a two dimensional index code. For this two-sided SUICP(SNCS), we have $K_a=11$  and $\Delta_a=2$. The AIR matrix of size $\mathbf{L}_{11 \times 9}$ is given below.
\begin{figure}[ht]
\arraycolsep=1pt
{
$$\mathbf{L}_{11 \times 9}=\left[\begin{array}{*{20}c}
   1 & 0 & 0 & 0 & 0 & 0 & 0 & 0 & 0\\
   0 & 1 & 0 & 0 & 0 & 0 & 0 & 0 & 0\\
   0 & 0 & 1 & 0 & 0 & 0 & 0 & 0 & 0\\
   0 & 0 & 0 & 1 & 0 & 0 & 0 & 0 & 0\\
   0 & 0 & 0 & 0 & 1 & 0 & 0 & 0 & 0\\
   0 & 0 & 0 & 0 & 0 & 1 & 0 & 0 & 0\\
   0 & 0 & 0 & 0 & 0 & 0 & 1 & 0 & 0\\
   0 & 0 & 0 & 0 & 0 & 0 & 0 & 1 & 0\\
   0 & 0 & 0 & 0 & 0 & 0 & 0 & 0 & 1\\
   1 & 0 & 1 & 0 & 1 & 0 & 1 & 0 & 1\\
   0 & 1 & 0 & 1 & 0 & 1 & 0 & 1 & 1\\
   \end{array}\right]$$
}

\end{figure}

The scalar linear index code for the one-sided SUICP(SNCS) with $K_a=11,\Delta_a=2$ is given by 
\begin{align*}
\mathfrak{C}=\{&y_0+y_9, ~~~y_1+y_{10},~~~y_2+y_9,~~~y_3+y_{10}, \\&y_4+y_9, ~~~y_5+y_{10}, ~~~y_6+y_9,~~~y_7+y_{10}, \\&y_8+y_9+y_{10}\}. 
\end{align*}

The VLIC for the given two-sided SUICP(SNCS) is obtained by replacing $y_s$ in $\mathfrak{C}$ with $x_{2s,1}+x_{2s+1,1}+x_{2(s-1),2}+x_{2(s-1)+1,2}$ for $s \in [0:10]$. The two dimensional VLIC for the is given SUICP(SNCS) is given in Table \ref{vcex2table1}.
\begin{table*}[ht]
\centering
\setlength\extrarowheight{0.5pt}
\begin{tabular}{|c|c|}
\hline
\textbf{$c_0$} & $\underbrace{x_{0,1}+x_{1,1}+x_{20,2}+x_{21,2}}_{y_0}+\underbrace{x_{18,1}+x_{19,1}+x_{16,2}+x_{17,2}}_{y_9}$\\
\hline
\textbf{$c_1$} & $\underbrace{x_{2,1}+x_{3,1}+x_{0,2}+x_{1,2}}_{y_1}+\underbrace{x_{20,1}+x_{21,1}+x_{0,2}+x_{1,2}}_{y_{10}}$ \\
\hline 
\textbf{$c_2$} & $\underbrace{x_{4,1}+x_{5,1}+x_{2,2}+x_{3,2}}_{y_2}+\underbrace{x_{18,1}+x_{19,1}+x_{16,2}+x_{17,2}}_{y_9}$  \\
\hline
\textbf{$c_3$} & $\underbrace{x_{6,1}+x_{7,1}+x_{4,2}+x_{5,2}}_{y_3}+\underbrace{x_{20,1}+x_{21,1}+x_{0,2}+x_{1,2}}_{y_{10}}$ \\
\hline
\textbf{$c_4$} & $\underbrace{x_{8,1}+x_{9,1}+x_{6,2}+x_{7,2}}_{y_4}+\underbrace{x_{18,1}+x_{19,1}+x_{16,2}+x_{17,2}}_{y_{9}}$ \\
\hline
\textbf{$c_5$} & $\underbrace{x_{10,1}+x_{11,1}+x_{8,2}+x_{9,2}}_{y_5}+\underbrace{x_{20,1}+x_{21,1}+x_{0,2}+x_{1,2}}_{y_{10}}$ \\
\hline
\textbf{$c_6$} & $\underbrace{x_{12,1}+x_{13,1}+x_{10,2}+x_{11,2}}_{y_6}+\underbrace{x_{18,1}+x_{19,1}+x_{16,2}+x_{17,2}}_{y_9}$ \\
\hline
\textbf{$c_7$} & $\underbrace{x_{14,1}+x_{15,1}+x_{12,2}+x_{13,2}}_{y_7}+\underbrace{x_{20,1}+x_{21,1}+x_{0,2}+x_{1,2}}_{y_{10}}$ \\
\hline 
\textbf{$c_8$} & $\underbrace{x_{16,1}+x_{17,1}+x_{14,2}+x_{15,2}}_{y_8}+\underbrace{x_{18,1}+x_{19,1}+x_{16,2}+x_{17,2}}_{y_9}+
\underbrace{x_{20,1}+x_{21,1}+x_{0,2}+x_{1,2}}_{y_{10}}$ \\
\hline 
\end{tabular}
\caption{Vector linear index code for the two-sided SUICP(SNCS) given in Example \ref{vcex2}}
\label{vcex2table1}
\end{table*}

Let $\tau_s$ be the set of broadcast symbols required to decode $y_s$ from $\mathfrak{C}$. The values of $\tau_s$ and $S_s$ for $s \in [0:10]$ are given in Table \ref{vctableex14}.
\begin{table}[ht]
\centering
\begin{tabular}{|c|c|c|}
\hline
{$\mathcal{W}_s$} & $\tau_s$ & $S_s$ \\
\hline
$y_0$ &$y_0+y_9,y_2+y_9$& $S_0=y_0+y_3$ \\
\hline
$y_1$&$y_1+y_{10},y_3+y_{10}$&  $S_1=y_1+y_3$ \\
\hline 
$y_2$ &$y_2+y_9,y_4+y_9$&$S_2=y_2+y_4$  \\
\hline
$y_3$ &$y_3+y_{10},y_5+y_{10}$&$S_3=y_3+y_5$  \\
\hline
$y_4$ &$y_4+y_9,y_6+y_9$& $S_4=y_4+y_6$  \\
\hline
$y_5$ &$y_5+y_{10},y_7+y_{10}$&$S_5=y_5+y_7$ \\
\hline
$y_6$ &$y_6+y_9,y_7+y_{10},$&$S_6=y_6+y_7+y_8$ \\
&$y_8+y_9+y_{10}$&\\
\hline
$y_7$ &$y_7+y_{10},y_8+y_9+y_{10}$&$S_7=y_7+y_8+y_9 $\\
\hline 
$y_8$ &$y_8+y_9+y_{10}$&$S_8=y_8+y_9+y_{10}$ \\
\hline
$y_9$ &$y_0+y_9$&$S_9=y_0+y_9$ \\
\hline
$y_{10}$ &$y_1+y_{10}$&$S_{10}=y_1+y_{10} $\\
\hline 
\end{tabular}
\caption{Decoding at each receiver for one-sided SUICP(SNCS) with $K_a=11,\Delta_a=2$}
\label{vctableex14}
\end{table}

Receiver $R_k$ wants to decode $x_{k,1}$ and $x_{k,2}$. $R_k$ first decode $x_{k,2}$ from $S_{s+1}$, then decodes $x_{k,1}$ from $S_s$ for every $k \in [0:21]$. The receiver $R_{14}$ decodes $x_{14,2}$ from $S_8$ and decodes $x_{14,1}$ from $S_7.$ From Table \ref{vctableex14}, we have 
\begin{align*}
&S_8=y_8+y_9+y_{10}=\underbrace{x_{16,1}+x_{17,1}+x_{14,2}+x_{15,2}}_{y_8}+\\&\underbrace{x_{18,1}+x_{19,1}+x_{16,2}+x_{17,2}}_{y_9}+
\underbrace{x_{20,1}+x_{21,1}+x_{18,2}+x_{19,2}}_{y_{10}}.
\end{align*} 

In $S_8$, all other messages except $x_{14,2}$ are in the side information of $R_{14}$. Hence $R_{14}$ decodes $x_{14,2}$ from $S_8$. From Table \ref{vctableex14}, we have 
\begin{align*}
&S_7=y_7+y_8+y_9=\underbrace{x_{14,1}+x_{15,1}+x_{12,2}+x_{13,2}}_{y_7}+\\&\underbrace{x_{16,1}+x_{17,1}+x_{14,2}+x_{15,2}}_{y_8}+
\underbrace{x_{18,1}+x_{19,1}+x_{16,2}+x_{17,2}}_{y_9}.
\end{align*} 

In $S_7$, all other messages except $x_{14,1}$ and $x_{14,2}$ are in the side information of $R_{14}$. $R_{14}$ already decoded $x_{14,2}$ and hence it cancels the interference in $S_7$ due to $x_{14,2}$ and then decodes $x_{14,1}$. 

The capacity of the given two-sided SUICP(SNCS) is $\frac{4}{18}=\frac{2}{9}$. The VLIC $\mathfrak{C}$ uses nine broadcast symbols to transmit two message symbols per receiver. The rate achieved by $\mathfrak{C}$ is $\frac{2}{9}$.
\end{example}
\begin{note}
For the two-sided SUICP(SNCS) given in Example \ref{vcex2}, the construction given in \cite{VaR1} gives four dimensional optimal length VLIC, whereas the construction given in this paper gives two dimensional optimal length VLIC.
\end{note}
\begin{example}
\label{vcex3}
Consider the two-sided SUICP(SNCS) with $K=24,D=11$ and $U=2$. For this SUICP(SNCS), we have $\text{gcd}(K,D-U,U+1)=3$. The optimal length index code for this SUICP(SNCS) is scalar linear index codes. For this two-sided SUICP(SNCS), we have $K_a=8,\Delta_a=3$. AIR matrix of size $8 \times 5$ is given below.
\arraycolsep=1pt
\setlength\extrarowheight{-2.0pt}
{
$$\mathbf{L}_{8 \times 5}=\left[\begin{array}{*{20}c}
   1 & 0 & 0 & 0 & 0 \\
   0 & 1 & 0 & 0 & 0 \\
   0 & 0 & 1 & 0 & 0 \\
   0 & 0 & 0 & 1 & 0 \\
   0 & 0 & 0 & 0 & 1 \\
   1 & 0 & 0 & 1 & 0 \\
   0 & 1 & 0 & 0 & 1 \\
   0 & 0 & 1 & 1 & 1 \\
  \end{array}\right]$$
}

The scalar linear code $\mathfrak{C}$ for the one-sided SUICP(SNCS) with $K_a=8$ and $\Delta_a=3$ is 
\begin{align*}
\mathfrak{C}&=\{c_0,c_1.c_2,c_3,c_4\}\\&=\{y_0+y_5, ~y_1+y_6, ~y_2+y_7, ~y_3+y_5+y_7, ~y_4+y_6+y_7\}.
\end{align*}
 
The scalar linear index code $\mathfrak{C}^{(2s)}$ for the given two-sided SUICP(SNCS) is given by
\begin{align*}
&\mathfrak{C}^{(2s)}=\{c_0,c_1,c_2,c_3,c_4\}=\\\{&\underbrace{x_0+x_1+x_2}_{y_0}+
\underbrace{x_{15}+x_{16}+x_{17}}_{y_5},\\&\underbrace{x_3+x_4+x_5}_{y_1}+\underbrace{x_{18}+x_{19}+x_{20}}_{y_6},\\&\underbrace{x_6+x_7+x_8}_{y_2}+\underbrace{x_{21}+x_{22}+x_{23}}_{y_7},\\& \underbrace{x_9+x_{10}+x_{11}}_{y_3}+\underbrace{x_{15}+x_{16}+x_{17}}_{y_5}+\underbrace{x_{21}+x_{22}+x_{23}}_{y_7},\\& \underbrace{x_{12}+x_{13}+x_{14}}_{y_4}+\underbrace{x_{18}+x_{19}+x_{20}}_{y_6}+\underbrace{x_{21}+x_{22}+x_{23}}_{y_7}\}.
\end{align*}
\end{example}
\begin{note}
For the two-sided SUICP(SNCS) given in Example \ref{vcex3}, the construction given in  \cite{VaR1} gives three dimensional optimal length VLIC, whereas the construction given in this paper gives  optimal length scalar linear index code.
\end{note}
\subsection{Encoding Matrices for Optimal Vector Linear Codes}
In this subsection, we give encoding matrices for optimal VLICs of two-sided SUICP(SNCS).  with $K$ messages, $D$ side-information after and $U$ side-information before the desired message. 

\begin{definition}
\label{vcdef5}
Let $\mathbf{L}$ be the AIR matrix of size $K_a \times (K_a-\Delta_a)$. Let $L_s$ be the $s$th row of $\mathbf{L}$ for $s \in [0:K_a-1]$. Then, the $s$th precoding matrix $\mathbf{V}_s$ is defined as 
$$\mathbf{V}_{s}=\left[\begin{array}{*{20}c}
   L_s  \\
   L_{s+1} \\
   L_{s+2} \\
    \vdots \\
   L_{s+u_a-1} \\
   \end{array}\right].$$
The matrix $\mathbf{\tilde{V}}_s$ is defined as 
\begin{align}
\label{precoding}
\mathbf{\tilde{V}}_{s}=\left.\left[\begin{array}{*{20}c}
   \mathbf{V}_s  \\
   \mathbf{V}_s  \\
   \vdots  \\
   \mathbf{V}_s 
   \end{array}\right]\right\rbrace a~\text{number~of}~ \mathbf{V}_s~\text{matrices}
\end{align}
\end{definition}
\begin{lemma}
\label{lemma1}
Consider the two-sided SUICP(SNCS) with $K$ messages, $U$ side-information after and $D$ side-information before the desired message. Let $x_k=(x_{k,1},x_{k,2},\cdots,x_{k,u_a})$ be the message vector wanted by the $k$th receiver, where $x_k \in \mathbb{F}_q^{u_a}, x_{k,i} \in \mathbb{F}_q$ for $k \in [0:K-1]$ and $i \in [1:u_a]$. For this two-sided SUICP(SNCS), an encoding matrix for optimal length VLIC $\mathbf{L}^{(2s)}$ of size $Ku_a \times (K_a-\Delta_a)$ is given by 
$$\mathbf{L}^{(2s)}=\left.\left[\begin{array}{*{20}c}
   \mathbf{\tilde{V}}_0  \\
   \mathbf{\tilde{V}}_1  \\
   \mathbf{\tilde{V}}_2  \\
   \vdots  \\
   \mathbf{\tilde{V}}_{K_a-1}   \\
   \end{array}\right]\right. $$
where $\mathbf{\tilde{V}}_k$ is defined in Definition \ref{vcdef5}.
\end{lemma}
\begin{proof}
Let $\mathbf{x}_k=[x_{k,1}~x_{k,2}~\ldots~x_{k,u_a}]$. Let $\mathbf{x}_{[a:b]}=[\mathbf{x}_a~\mathbf{x}_{a+1}~\ldots~\mathbf{x}_b]$. From Theorem \ref{vcthm1}, the $u_a$ dimensional VLIC for the given two-sided SUICP(SNCS) is given by 
\begin{align*}
[c_0~c_1\ldots &c_{K_a-\Delta_a-1}]=\sum_{s=0}^{K_a-1}y_sL_s\\&=\sum_{s=0}^{K_a-1} \left(\sum_{i=1}^{u_a} \sum_{j=0}^{a-1}x_{a(s+1-i)+j,i}\right)L_s \\&=\sum_{s=0}^{K_a-1} (\sum_{j=0}^{a-1} \left(\sum_{i=1}^{u_a}x_{a(s+1-i)+j,i} L_s\right) \\&=\sum_{s=0}^{K_a-1} \left(\sum_{j=0}^{a-1} \mathbf{x}_{as+j}\mathbf{V}_s \right) \\&= \sum_{s=0}^{K_a-1} \mathbf{x}_{[as:as+a-1]} \mathbf{\tilde{V}_s}= \mathbf{x}_{[0:K-1]}\mathbf{L}^{(2s)}.
\end{align*}

This completes the proof.
\end{proof}
\begin{example}
\label{vcex4}
Consider the two-sided SUICP(SNCS) with $K=8,D=2,U=1$. The optimal length VLIC for this SUICP is given in Example \ref{vcex1}. The encoding matrix $\mathbf{L}_{16 \times 7}^{(2s)}$ for this optimal length VLIC by using Lemma \ref{lemma1} is given in Fig \ref{encodingmatrix2}.
\end{example}
\begin{figure}
\arraycolsep=1pt
\setlength\extrarowheight{-2.0pt}
{
$$\mathbf{L}_{16 \times 7}^{(2s)}=\left[\begin{array}{*{20}c}
   1 & 0 & 0 & 0 & 0 & 0 & 0\\
   0 & 1 & 0 & 0 & 0 & 0 & 0\\
   0 & 1 & 0 & 0 & 0 & 0 & 0\\
   0 & 0 & 1 & 0 & 0 & 0 & 0\\
   0 & 0 & 1 & 0 & 0 & 0 & 0\\
   0 & 0 & 0 & 1 & 0 & 0 & 0\\
   0 & 0 & 0 & 1 & 0 & 0 & 0\\
   0 & 0 & 0 & 0 & 1 & 0 & 0\\
   0 & 0 & 0 & 0 & 1 & 0 & 0\\
   0 & 0 & 0 & 0 & 0 & 1 & 0\\
   0 & 0 & 0 & 0 & 0 & 1 & 0\\
   0 & 0 & 0 & 0 & 0 & 0 & 1\\
   0 & 0 & 0 & 0 & 0 & 0 & 1\\
   1 & 1 & 1 & 1 & 1 & 1 & 1\\
   1 & 1 & 1 & 1 & 1 & 1 & 1\\
   1 & 0 & 0 & 0 & 0 & 0 & 0\\
   \end{array}\right]$$
}
\caption{Encoding matrices for two-sided SUICP(SNCS) given in Example \ref{vcex4}}
\label{encodingmatrix2}
\end{figure}
\begin{example}
\label{vcex5}
Consider the two-sided SUICP(SNCS) with $K=22,D=7,U=3$. The optimal length VLIC for this SUICP is given in Example \ref{vcex2}. The encoding matrix $\mathbf{L}_{44 \times 9}^{(2s)}$ for this optimal length VLIC by using Lemma \ref{lemma1} is given in Fig \ref{encodingmatrix}.
\begin{figure}
\arraycolsep=1pt
\setlength\extrarowheight{-2.0pt}
{
$$\mathbf{L}_{44 \times 9}^{(2s)}=\left[\begin{array}{*{20}c}
   1 & 0 & 0 & 0 & 0 & 0 & 0 & 0 & 0\\
   0 & 1 & 0 & 0 & 0 & 0 & 0 & 0 & 0\\
   1 & 0 & 0 & 0 & 0 & 0 & 0 & 0 & 0\\
   0 & 1 & 0 & 0 & 0 & 0 & 0 & 0 & 0\\
   0 & 1 & 0 & 0 & 0 & 0 & 0 & 0 & 0\\
   0 & 0 & 1 & 0 & 0 & 0 & 0 & 0 & 0\\
   0 & 1 & 0 & 0 & 0 & 0 & 0 & 0 & 0\\
   0 & 0 & 1 & 0 & 0 & 0 & 0 & 0 & 0\\
   0 & 0 & 1 & 0 & 0 & 0 & 0 & 0 & 0\\
   0 & 0 & 0 & 1 & 0 & 0 & 0 & 0 & 0\\
   0 & 0 & 1 & 0 & 0 & 0 & 0 & 0 & 0\\
   0 & 0 & 0 & 1 & 0 & 0 & 0 & 0 & 0\\
   
   0 & 0 & 0 & 1 & 0 & 0 & 0 & 0 & 0\\
   0 & 0 & 0 & 0 & 1 & 0 & 0 & 0 & 0\\
   0 & 0 & 0 & 1 & 0 & 0 & 0 & 0 & 0\\
   0 & 0 & 0 & 0 & 1 & 0 & 0 & 0 & 0\\
   
   0 & 0 & 0 & 0 & 1 & 0 & 0 & 0 & 0\\
   0 & 0 & 0 & 0 & 0 & 1 & 0 & 0 & 0\\
   0 & 0 & 0 & 0 & 1 & 0 & 0 & 0 & 0\\
   0 & 0 & 0 & 0 & 0 & 1 & 0 & 0 & 0\\
   
   0 & 0 & 0 & 0 & 0 & 1 & 0 & 0 & 0\\
   0 & 0 & 0 & 0 & 0 & 0 & 1 & 0 & 0\\
   0 & 0 & 0 & 0 & 0 & 1 & 0 & 0 & 0\\
   0 & 0 & 0 & 0 & 0 & 0 & 1 & 0 & 0\\
   
   0 & 0 & 0 & 0 & 0 & 0 & 1 & 0 & 0\\
   0 & 0 & 0 & 0 & 0 & 0 & 0 & 1 & 0\\
   0 & 0 & 0 & 0 & 0 & 0 & 1 & 0 & 0\\
   0 & 0 & 0 & 0 & 0 & 0 & 0 & 1 & 0\\
   
   0 & 0 & 0 & 0 & 0 & 0 & 0 & 1 & 0\\
   0 & 0 & 0 & 0 & 0 & 0 & 0 & 0 & 1\\
   0 & 0 & 0 & 0 & 0 & 0 & 0 & 1 & 0\\
   0 & 0 & 0 & 0 & 0 & 0 & 0 & 0 & 1\\
   
   0 & 0 & 0 & 0 & 0 & 0 & 0 & 0 & 1\\
   1 & 0 & 1 & 0 & 1 & 0 & 1 & 0 & 1\\
   0 & 0 & 0 & 0 & 0 & 0 & 0 & 0 & 1\\
   1 & 0 & 1 & 0 & 1 & 0 & 1 & 0 & 1\\
   
   1 & 0 & 1 & 0 & 1 & 0 & 1 & 0 & 1\\
   0 & 1 & 0 & 1 & 0 & 1 & 0 & 1 & 1\\
   1 & 0 & 1 & 0 & 1 & 0 & 1 & 0 & 1\\
   0 & 1 & 0 & 1 & 0 & 1 & 0 & 1 & 1\\
   
   0 & 1 & 0 & 1 & 0 & 1 & 0 & 1 & 1\\
   1 & 0 & 0 & 0 & 0 & 0 & 0 & 0 & 0\\
   0 & 1 & 0 & 1 & 0 & 1 & 0 & 1 & 1\\
   1 & 0 & 0 & 0 & 0 & 0 & 0 & 0 & 0\\
   \end{array}\right]~\mathbf{L}_{24 \times 5}^{(2s)}=\left[\begin{array}{*{20}c}
   1 & 0 & 0 & 0 & 0  \\
   1 & 0 & 0 & 0 & 0  \\
   1 & 0 & 0 & 0 & 0  \\
   0 & 1 & 0 & 0 & 0  \\
   0 & 1 & 0 & 0 & 0  \\
   0 & 1 & 0 & 0 & 0  \\
   0 & 0 & 1 & 0 & 0  \\
   0 & 0 & 1 & 0 & 0  \\
   0 & 0 & 1 & 0 & 0  \\
   0 & 0 & 0 & 1 & 0  \\
   0 & 0 & 0 & 1 & 0  \\
   0 & 0 & 0 & 1 & 0  \\
   0 & 0 & 0 & 0 & 1  \\
   0 & 0 & 0 & 0 & 1  \\
   0 & 0 & 0 & 0 & 1  \\
   1 & 0 & 0 & 1 & 0  \\
   1 & 0 & 0 & 1 & 0  \\
   1 & 0 & 0 & 1 & 0  \\
   0 & 1 & 0 & 0 & 1  \\
   0 & 1 & 0 & 0 & 1  \\
   0 & 1 & 0 & 0 & 1  \\
   0 & 0 & 1 & 1 & 1  \\
   0 & 0 & 1 & 1 & 1  \\
   0 & 0 & 1 & 1 & 1  \\
   \end{array}\right]$$
}
\caption{Encoding matrices for two-sided SUICP(SNCS) given in Example \ref{vcex5} and \ref{vcex6}}
\label{encodingmatrix}
\end{figure}
\end{example}
\begin{example}
\label{vcex6}
Consider the two-sided SUICP(SNCS) with $K=24,D=11,U=2$. The optimal length VLIC for this SUICP is given in Example \ref{vcex3}. The encoding matrix $\mathbf{L}_{24 \times 5}^{(2s)}$ for this optimal length VLIC by using Lemma \ref{lemma1} is given in Fig \ref{encodingmatrix}.
\end{example}
\section{Some special cases of proposed construction}
\label{sec3}
In this section, we give four special cases of the proposed construction. Consider the two-sided SUICP(SNCS) with $K$ messages, $D$ side-information after and $U$ side-information before the desired message as given in \eqref{antidote}. We have $a=\text{gcd}(K,D-U,U+1)$. In Theorem \ref{vcthm1}, we combine $U+1$ message symbols, one from each of the message vectors $x_t$ for every $t \in [as-(U+1)+a:as+a-1]$ as given below.
\begin{align}
\label{thm2eq6}
\nonumber
y_s=\sum_{i=1}^{u_a} \sum_{j=0}^{a-1}& x_{a(s+1-i)+j,i} \\& \text{for}~s \in [0:K_a-1].
\end{align}
\subsection{Optimal scalar linear index codes for two-sided SUICP(SNCS)}
In the construction given in Theorem \ref{vcthm1}, if $U+1$ divides both $K$ and $D-U$, then we have $a=U+1$ and $u_a=1$. In this case, \eqref{thm2eq6} reduces to 
\begin{align}
\label{thm2eq7}
\nonumber
y_s&=\sum_{i=1}^1 \sum_{j=0}^U x_{(U+1)(s+1-i)+j,i}\\&=\sum_{j=0}^U x_{(U+1)s+j} ~~~ \text{for}~s \in [0:K_a].
\end{align}

Hence, in this case, the index code given by \eqref{thm1eq5} is an optimal length scalar linear index code.
\begin{example}
\label{vcex7}
Consider the two-sided SUICP(SNCS) with $K=14,D=3$ and $U=1$. For this SUICP(SNCS), we have $\text{gcd}(K,D-U,U+1)=U+1=2$. The optimal length index code for this two-sided SUICP(SNCS) is a scalar linear index code. For this two-sided SUICP(SNCS), we have $K_a=7,\Delta_a=1$. An optimal length encoding matrix for this one-sided SUICP(SNCS) is a AIR matrix of size $7 \times 6$. AIR matrix of size $7 \times 6$ is given below.
\arraycolsep=1pt
\setlength\extrarowheight{-2.0pt}
{
$$\mathbf{L}_{7 \times 6}=\left[\begin{array}{*{20}c}
   1 & 0 & 0 & 0 & 0 & 0\\
   0 & 1 & 0 & 0 & 0 & 0\\
   0 & 0 & 1 & 0 & 0 & 0\\
   0 & 0 & 0 & 1 & 0 & 0\\
   0 & 0 & 0 & 0 & 1 & 0\\
   0 & 0 & 0 & 0 & 0 & 1\\
   1 & 1 & 1 & 1 & 1 & 1\\
  \end{array}\right]$$
}

The scalar linear index code $\mathfrak{C}$ for this one-sided SUICP(SNCS) is given below.
\begin{align*}
\mathfrak{C}=\{c_0,c_1,c_2,c_3,c_4,c_5\}=\{&y_0+y_6,~~~y_1+y_6,\\&y_2+y_6,~~~y_3+y_6,\\&y_4+y_6,~~~y_5+y_6\}.
\end{align*}
 
The scalar linear index code $\mathfrak{C}^{(2s)}$ for the given two-sided SUICP(SNCS) is obtained by replacing $y_s$ in $\mathfrak{C}$ with $x_{2s}+x_{2s+1}$ for $s \in [0:6]$. The scalar linear index code $\mathfrak{C}^{(2s)}$ is given below.
\begin{align*}
\mathfrak{C}^{(2s)}=\{&x_0+x_1+x_{12}+x_{13},~~~x_2+x_3+x_{12}+x_{13},\\&x_4+x_5+x_{12}+x_{13},~~~~x_6+x_7+x_{12}+x_{13},\\&x_8+x_9+x_{12}+x_{13},~~~x_{10}+x_{11}+x_{12}+x_{13}\}.
\end{align*}
\end{example}
\subsection{Optimal $(U+1)$ dimensional VLICs for two-sided SUICP(SNCS)}
In the construction given in Theorem \ref{vcthm1}, if $\text{gcd}(K,D-U,U+1)=1$, then we have $a=1$ and $u_a=U+1$. In this case, \eqref{thm2eq6} reduces to 
\begin{align}
\label{thm2eq8}
\nonumber
y_s&=\sum_{i=1}^{U+1} \sum_{j=0}^0 x_{(s+1-i)+j,i}\\&=\sum_{i=1}^{U+1} x_{s+1-i,i} ~~\text{for}~s \in [0:K-1].
\end{align}

Hence, in this case, the index code given by \eqref{thm1eq5} is an optimal length $(U+1)$ dimensional VLIC.

\begin{example}
\label{vcex8}
Consider the two-sided SUICP(SNCS) with $K=11, D=5, U=2$. For this SUICP(SNCS), $\text{gcd}(K,D-U,U+1)=a=1$. The optimal length index code is a three dimensional VLIC. For this SUICP(SNCS), we have $K_a=K=11$  and $\Delta_a=3$. The AIR matrix of size $\mathbf{L}_{11 \times 8}$ is given below.

\begin{small}
\arraycolsep=1pt
\setlength\extrarowheight{-2.0pt}
{
$$\mathbf{L}_{11 \times 8}=\left[\begin{array}{*{20}c}
   1 & 0 & 0 & 0 & 0 & 0 & 0 & 0\\
   0 & 1 & 0 & 0 & 0 & 0 & 0 & 0\\
   0 & 0 & 1 & 0 & 0 & 0 & 0 & 0\\
   0 & 0 & 0 & 1 & 0 & 0 & 0 & 0\\
   0 & 0 & 0 & 0 & 1 & 0 & 0 & 0\\
   0 & 0 & 0 & 0 & 0 & 1 & 0 & 0\\
   0 & 0 & 0 & 0 & 0 & 0 & 1 & 0\\
   0 & 0 & 0 & 0 & 0 & 0 & 0 & 1\\
   1 & 0 & 0 & 1 & 0 & 0 & 1 & 0\\
   0 & 1 & 0 & 0 & 1 & 0 & 0 & 1\\
   0 & 0 & 1 & 0 & 0 & 1 & 1 & 1\\
   \end{array}\right]$$
}
\end{small}

The scalar linear index code $\mathfrak{C}$ for the one-sided SUICP(SNCS) with $K_a=11$ and $\Delta_a=3$ is given below. 
\begin{align*}
\mathfrak{C}=\{&y_0+y_8, ~~y_1+y_9,~~y_2+y_{10},\\&y_3+y_8, ~~y_4+y_9, ~~y_5+y_{10},\\& y_6+y_8+y_{10}, ~y_7+y_9+y_{10}\}. 
\end{align*}

The VLIC $\mathfrak{C}^{(2s)}$ for the given two-sided SUICP(SNCS) is obtained by replacing $y_s$ in $\mathfrak{C}$ with $x_{s,1}+x_{s-1,2}+x_{s-2,3}$ for $s \in [0:10]$. The VLIC $\mathfrak{C}^{(2s)}$ for the given two-sided SUICP(SNCS) is given below.
\begin{align*}
\mathfrak{C}^{(2s)}=\{&\underbrace{x_{0,1}+x_{10,2}+x_{9,3}}_{y_0}+\underbrace{x_{8,1}+x_{7,2}+x_{6,3}}_{y_8},\\& \underbrace{x_{1,1}+x_{0,2}+x_{10,3}}_{y_1}+\underbrace{x_{9,1}+x_{8,2}+x_{7,3}}_{y_9}, \\& \underbrace{x_{2,1}+x_{1,2}+x_{0,3}}_{y_2}+\underbrace{x_{10,1}+x_{9,2}+x_{8,3}}_{y_{10}},\\& \underbrace{x_{3,1}+x_{2,2}+x_{1,3}}_{y_3}+\underbrace{x_{8,1}+x_{7,2}+x_{6,3}}_{y_8}, \\& \underbrace{x_{4,1}+x_{3,2}+x_{2,3}}_{y_4}+\underbrace{x_{9,1}+x_{8,2}+x_{7,3}}_{y_9},\\& \underbrace{x_{5,1}+x_{4,2}+x_{3,3}}_{y_5}+\underbrace{x_{10,1}+x_{9,2}+x_{8,3}}_{y_{10}}, \\& \underbrace{x_{6,1}+x_{5,2}+x_{4,3}}_{y_6}+\underbrace{x_{8,1}+x_{7,2}+x_{6,3}}_{y_8}\\&+\underbrace{x_{10,1}+x_{9,2}+
x_{8,3}}_{y_{10}},\\& \underbrace{x_{7,1}+x_{6,2}+x_{5,3}}_{y_7}+\underbrace{x_{9,1}+x_{8,2}+x_{7,3}}_{y_9}\\& +\underbrace{x_{10,1}+x_{9,2}+
x_{9,3}}_{y_{10}}\}.
\end{align*}

\end{example}
\subsection{Optimal scalar linear index codes for one-sided SUICP(SNCS)}
In the construction given in Theorem \ref{vcthm1}, if $U=0$, then we have, $a=\text{gcd}(K,D-U,U+1)=1$ and $u_a=U+1=1$. In this case, \eqref{thm2eq6} reduces to 
$y_s=x_s ~~\text{for}~s \in [0:K-1].$

Hence, in this case, the index code given by \eqref{thm1eq5} is an optimal length scalar linear index code.
\begin{example}
\label{vcex9}
Consider the two-sided SUICP(SNCS) with $K=7,D=2$ and $U=0$. For this SUICP(SNCS), we have $\text{gcd}(K,D-U,U+1)=1$. The optimal length index code for this SUICP(SNCS) is scalar linear index code. For this SUICP(SNCS), we have $K=K_a=7,D=\Delta_a=2$. An optimal length encoding matrix for this SUICP(SNCS) is a AIR matrix of size $7 \times 5$. AIR matrix of size $7 \times 5$ is given below.
\arraycolsep=1pt
\setlength\extrarowheight{-2.0pt}
{
$$\mathbf{L}_{7 \times 5}=\left[\begin{array}{*{20}c}
   1 & 0 & 0 & 0 & 0 \\
   0 & 1 & 0 & 0 & 0 \\
   0 & 0 & 1 & 0 & 0 \\
   0 & 0 & 0 & 1 & 0 \\
   0 & 0 & 0 & 0 & 1 \\
   1 & 0 & 1 & 0 & 1 \\
   0 & 1 & 0 & 1 & 1 \\
  \end{array}\right]$$
}

The scalar linear index code $\mathfrak{C}^{(2s)}$ for this SUICP(SNCS) is given below.
\begin{align*}
\mathfrak{C}^{(2s)}=\mathfrak{C}=\{&x_0+x_5,~~~x_1+x_6,~~~x_2+x_5,\\&x_3+x_6,~~~x_4+x_5+x_6\}.
\end{align*}
\end{example}
\subsection{Optimal $\frac{U+1}{D-U}$ dimensional VLICs for two-sided SUICP(SNCS)}
In the construction given in Theorem \ref{vcthm1}, if $\text{gcd}(K,D-U,U+1)=D-U,$ then we have $u_a=\frac{U+1}{D-U}$, $K_a=\frac{K}{D-U}$ and $\Delta_a=1$. Hence, in this case, the AIR matrix used is always of dimension $K_a \times (K_a-1)$ and the index code given by \eqref{thm1eq5} is an optimal length $\frac{U+1}{D-U}$ dimensional VLIC. In \cite{TRCR}, it is shown that the message probability of error in decoding a message at a particular receiver decreases with a decrease in the number of transmissions used to decode the message among the total of broadcast transmissions. In this case, as the AIR matrix size is always of the form $K_a \times (K_a-1)$, the low complexity decoding given in \cite{VaR3} enables each receiver to decode its wanted message symbol by using atmost two broadcast symbols.

\begin{example}
\label{vcex10}
Consider the two-sided SUICP(SNCS) with $K=12,D=5$ and $U=3$. For this SUICP(SNCS), we have $\text{gcd}(K,D-U,U+1)=D-U=2$. The optimal length index code for this two-sided SUICP(SNCS) is a two dimensional VLIC. For this two-sided  SUICP(SNCS), we have $K_a=6,\Delta_a=1$. An optimal length encoding matrix for this one-sided SUICP(SNCS) is a AIR matrix of size $6 \times 5$. AIR matrix of size $6 \times 5$ is given below.
\arraycolsep=1pt
\setlength\extrarowheight{-2.0pt}
{
$$\mathbf{L}_{6 \times 5}=\left[\begin{array}{*{20}c}
   1 & 0 & 0 & 0 & 0 \\
   0 & 1 & 0 & 0 & 0 \\
   0 & 0 & 1 & 0 & 0 \\
   0 & 0 & 0 & 1 & 0 \\
   0 & 0 & 0 & 0 & 1 \\
   0 & 0 & 0 & 0 & 0 \\
   1 & 1 & 1 & 1 & 1 \\
  \end{array}\right]$$
}

The scalar linear index code $\mathfrak{C}$ for this one-sided SUICP(SNCS) is given below.
\begin{align*}
\mathfrak{C}=\{c_0,c_1.c_2,c_3,c_4\}=\{&y_0+y_5,~~~y_1+y_5,~~~y_2+y_5,\\&y_3+y_5,~~~y_4+y_5\}.
\end{align*}
 
The scalar linear index code $\mathfrak{C}^{(2s)}$ for the given two-sided SUICP(SNCS) is obtained by replacing $y_s$ in $\mathfrak{C}$ with $x_{2s,1}+x_{2s-2,2}+x_{2s+1,1}+x_{2s-1,2}$ for $s \in [0:5]$. 

\end{example}

\begin{example}
\label{vcex11}
Consider the two-sided SUICP(SNCS) with $K=27,D=8$ and $U=5$. For this SUICP(SNCS), we have $\text{gcd}(K,D-U,U+1)=D-U=3$. The optimal length index code for this two-sided SUICP(SNCS) is a two dimensional VLIC. For this two-sided SUICP(SNCS), we have $K_a=9,\Delta_a=1$. An optimal length encoding matrix for this one-sided SUICP(SNCS) is a AIR matrix of size $9 \times 8$. AIR matrix of size $9 \times 8$ is given below.
\arraycolsep=1pt
\setlength\extrarowheight{-2.0pt}
{
$$\mathbf{L}_{9 \times 8}=\left[\begin{array}{*{20}c}
   1 & 0 & 0 & 0 & 0 & 0 & 0 & 0\\
   0 & 1 & 0 & 0 & 0 & 0 & 0 & 0\\
   0 & 0 & 1 & 0 & 0 & 0 & 0 & 0\\
   0 & 0 & 0 & 1 & 0 & 0 & 0 & 0\\
   0 & 0 & 0 & 0 & 1 & 0 & 0 & 0\\
   0 & 0 & 0 & 0 & 0 & 1 & 0 & 0\\
   0 & 0 & 0 & 0 & 0 & 0 & 1 & 0\\
   0 & 0 & 0 & 0 & 0 & 0 & 0 & 1\\
   1 & 1 & 1 & 1 & 1 & 1 & 1 & 1\\
  \end{array}\right]$$
}

The scalar linear index code $\mathfrak{C}$ for this one-sided SUICP(SNCS) is given below.
\begin{align*}
\mathfrak{C}=\{&y_0+y_8,~~~y_1+y_8,~~~y_2+y_8,~~~y_3+y_8,\\&y_4+y_8,~~~y_5+y_8,~~~y_6+y_8,~~~y_7+y_8\}.
\end{align*}
 
The two dimensional VLIC $\mathfrak{C}^{(2s)}$ for the given two-sided SUICP(SNCS) is obtained by replacing $y_s$ in $\mathfrak{C}$ with $x_{3s,1}+x_{3s-3,2}+x_{3s+1,1}+x_{3s-2,2}+x_{3s+2,1}+x_{3s-1,2}$ for $s \in [0:8]$. 

\end{example}
\section*{Acknowledgment}
This work was supported partly by the Science and Engineering Research Board (SERB) of Department of Science and Technology (DST), Government of India, through J.C. Bose National Fellowship to B. Sundar Rajan.


\begin{thebibliography}{160}

\bibitem{MCJ}
H. Maleki, V. Cadambe, and S. Jafar, ``Index coding – an interference alignment perspective", in IEEE \textit{Trans. Inf. Theory,}, vol. 60, no.9, pp.5402-5432, Sep. 2014.

\bibitem{TIM}
S. A. Jafar, “Topological interference management through index coding,” IEEE Trans. Inf. Theory, vol. 60, no. 1, pp. 529–568, Jan. 2014.

\bibitem{BiK}
Y. Birk and T. Kol, ``Coding-on-demand by an informed-source (ISCOD) for efficient broadcast of different supplemental data to caching clients", in IEEE \textit{Trans. Inf. Theory,}, vol. 52, no.6, pp.2825-2830, June 2006.

\bibitem{OnH}
L Ong and C K Ho, ``Optimal Index Codes for a Class of Multicast Networks with Receiver Side-information'', in \textit{Proc. IEEE ICC}, 2012, pp. 2213-2218.

\bibitem{YBJK}
Z. Bar-Yossef, Z. Birk, T. S. Jayram, and T. Kol, ``Index coding with side-information", in IEEE \textit{Trans. Inf. Theory,}, vol. 57, no.3, pp.1479-1494, Mar. 2011.

\bibitem{TRCR}
A. Thomas, R. Kavitha, A. Chandramouli and B. S. Rajan,
``Single uniprior index coding with min-max probability of error over fading channels,"
in \textit{Proc. IEEE Trans. on Veh. Technology} July 2017. 

\bibitem{VaR1}
M. B. Vaddi and B. S. Rajan, ``Optimal vector linear index codes for some symmetric multiple unicast problems,'' in \textit{Proc. IEEE ISIT}, Barcelona, Spain July 2016, pp. 125-129.

\bibitem{VaR2}
M. B. Vaddi and B. S. Rajan, ``Optimal scalar linear index codes for one-sided neighboring side-information problems,'' \textit{In Proc. IEEE GLOBECOM Workshop on Network Coding and Applications}, Washington, USA, December 2016.

\bibitem{VaR3}
M. B. Vaddi and B. S. Rajan, ``Low-complexity decoding for symmetric, neighboring and consecutive side-information index coding problems,'' in arXiv:1705.03192v2 [cs.IT] 16 May 2017.


\end{thebibliography}
\end{document}